\documentclass{article}
\usepackage{soul}
\usepackage{url}
\usepackage[hidelinks]{hyperref}
\usepackage[utf8]{inputenc}
\usepackage[small]{caption}
\usepackage{graphicx}
\usepackage{amsmath}
\usepackage{amsthm}
\usepackage{booktabs}
\usepackage{algorithmic}
\usepackage[switch]{lineno}
\usepackage{appendix}
\usepackage{graphicx} 
\usepackage{xcolor}

\usepackage{graphicx,fancyhdr,amsmath,amssymb,amsthm,subfig,url,hyperref,setspace,apxproof,mathtools,array,multirow,multicol,makecell,authblk}
\usepackage[linesnumbered,ruled]{algorithm2e}

\newtheorem{theorem}{Theorem}

\newtheorem{conjecture}[theorem]{Conjecture}

\newtheorem{definition}{Definition}

\newtheorem{lemma}[theorem]{Lemma}

\newtheorem{remark}{Remark}

\DeclarePairedDelimiter\floor{\lfloor}{\rfloor}

\newcommand{\OPT}{\mathsf{OPT}}
\newcommand{\PoND}{\mathsf{PoND}}

\title{Multi-Agent Non-Discriminatory Contracts}

\author[1]{Ke Ding\thanks{coco-ke.ding@connect.polyu.hk}}
\author[1]{Bo Li\thanks{comp-bo.li@polyu.edu.hk}}
\author[2]{Ankang Sun\thanks{sunankang1995@gmail.com}}

\affil[1]{Department of Computing, The Hong Kong Polytechnic University}
\affil[2]{School of Computer Science and Technology, Shandong University}

\begin{document}

\maketitle

\begin{abstract}




    We study multi-agent contracts, in which a principal delegates a task to multiple agents and incentivizes them to exert effort.
    Prior research has mostly focused on maximizing the principal’s utility, often resulting in highly disparate payments among agents. 
    Such disparities among agents may be undesirable in practice, for example, in standardized public contracting or worker cooperatives where fairness concerns are essential.
    Motivated by these considerations, our objective is to quantify the tradeoff between maximizing the principal's utility and equalizing payments among agents, which we call the {\em price of non-discrimination}. 
    Our first result is an almost tight bound on the price of non-discrimination, which scales logarithmically with the number of agents. 
    This bound can be improved to a constant by allowing some relaxation of the non-discrimination requirement. 
    We then provide a comprehensive characterization of the tradeoff between the level of non-discrimination and the loss in the optimal utility. 
\end{abstract}

\section{Introduction} 


As one of the pillars in microeconomic theory, contract theory utilizes the principal-agent model to achieve optimal outcomes amidst asymmetric information \cite{holmstrom1979moral,holmstrom1982moral}. It focuses on designing the market for services, analogous to how mechanism design and auction theory structure the market for goods. The past few years have witnessed growing interest in contract theory in both the economics and computation, with the applications extending to real-world scenarios like crowdsourcing platforms \cite{ho2014adaptive}, social media \cite{2021socialmedia}, and government-run programs \cite{dai2021contract}.


In the vanilla hidden-action principal-agent model \cite{holmstrom1979moral,grossman1992analysis}, the principal delegates a single task to an agent who may take one of costly actions and its outcome determines the reward. A contract, structured as a monetary transfer, incentivizes the agent to take the action maximizing his utility, which equals the transfer minus the cost. The principal cannot observe the action taken by the agent but only the outcome. Her goal is also to maximize her utility, which equals the reward obtained minus the payment transferred.  

The basic model was later extended to the setting where the principal incentivizes a team of agents for a single project \cite{holmstrom1982moral,babaioff2006combinatorial}. In this framework, each of the agents faces a binary choice: to exert effort at a specific cost, or shirk without cost. This multi-agent model involves strategic interactions among agents, ensuring that the resulting game induces a pure Nash equilibrium. 
For any subset of agents, the best way to incentivize them is to pay each agent his indifference payment between exerting effort and shirking, while paying zero to others. 
The scope of recent works focuses on computing the optimal contract \cite{dutting2023multi,ezra2023approximability,dutting2025multi}.

In the multi-agent model, payments in the optimal contract may vary significantly across the agents. 
However, in practice, companies and organizations frequently implement standardized payment schemes. For example, chain stores such as McDonald's or Starbucks set standard hour rates to treat every worker equally. A strengthened form of this egalitarianism is found in worker cooperatives, a concept dating back to the English labor movement. In organizations like the UK wholesaler Suma Wholefoods, equal pay is implemented to reflect the view that every worker is a co-owner and all contributions are essential to the business's success \cite{2025ukstore}.

Although equal pay brings high cohesion and equality, it is not universally viable. Not only because high-skill workers can command higher wages elsewhere, but also those high performers consider it unfair. A practical trade-off is to adopt the wage ratio, which refers to the gap between the top salaries and the bottom salaries in economy. For example, as a world leader in the co-operative movement, the Mondragon Corporation in Spain uses agreed-upon wage ratios \cite{2022spaincoop}. The ratios range from 1:3 to 1:9 in its different cooperatives. The statutory minimum wage could also be viewed through a similar lens: by having the highest wage fixed, the minimum wage was guaranteed. 

In this paper, we study contracts that give agents equitable payments, as considered in \cite{feng2024price,castiglioni2025fairteamcontracts}. Motivated by wage ratios used in practice, we extend this by allowing agents' payments to differ by a multiplicative factor.
Our goal is to offer a comprehensive understanding of how the degree of non-discrimination impacts the principal's optimal utility.

\subsection{Contribution}
We study multi-agent contracts, where a principal delegates a task to multiple agents.
Each agent incurs a cost upon exerting effort, and the overall success probability across agents is monotone submodular. 
The principal incentivizes agents to contribute effort by offering contractual payments, aiming to maximizing her utility. 
Our paper specifically centers on non-discriminatory contracts, which ensure equitable payments to all agents, and we explore their impact on the principal's optimal utility.
The results are summarized in Table \ref{res}.

\renewcommand{\arraystretch}{1.8}
\newcolumntype{C}[1]{>{\centering\arraybackslash}m{#1}}
\begin{table*}[t]
\centering
\begin{tabular}{|C{1.5cm}|C{1.8cm}|C{4cm}|C{4cm}|}
\hline
\multicolumn{2}{|c|}{} & Upper Bound & Lower Bound \\
\hline
\multicolumn{2}{|c|}{$\PoND$} 
& $O(\log n)$ (Lemma \ref{nd_ub}) 
& $\Omega(\frac{\log n}{\log \log n})$ (Lemma \ref{nd_lb}) \\
\hline
& \makecell[c]{
\rule{0pt}{2.8ex}
$0<\delta<1$
} 
& \makecell[c]{
\rule{0pt}{2.8ex}
{$\left\lceil \frac{1}{\delta}\right\rceil+1$} (Lemma \ref{bnd_ub})
} 
& \makecell[c]{
\rule{0pt}{2.8ex}
$\max \{2, \frac{1}{\delta+o(1)}\}$\\
(Lemmas \ref{bnd_epslb} and \ref{bnd_2lb})
} \\
\cline{2-4}
$\PoND(\beta)$  & $\delta=1$ 
& $2$ (Lemma \ref{bnd_ub}) 
& $1.591$ (Lemma \ref{bnd_1.6lb}) \\
\cline{2-4}
($\beta=n^\delta$) & $\delta>1$  
& \multicolumn{2}{c|}{$1+o(1)$ (Remark \ref{r1})}  \\
\cline{2-4}
& $n=2$  
& \multicolumn{2}{c|}{$1 + \frac{1}{\sqrt{\beta+1}}$ (Theorem \ref{bnd:2a})}   \\
\hline
\end{tabular}

\caption{Summary of results for the price of ($\beta$-)non-discrimination, where $n$ is the number of agents and $\beta$ is the payment gap between the linear contract.}
\label{res}

\end{table*}

Similar to \cite{feng2024price}, we consider the concept of price of non-discrimination ($\PoND$), which is defined as the supremum ratio between the principal's optimal utility of unconstrained contracts and that of all non-discriminatory contracts.
For any submodular reward function with $n$ agents, we design a non-discriminatory contract that achieves $\Omega(\frac{1}{\log n})$ fraction of the optimal utility without constraints\footnote{Throughout the paper, $\log$ denotes the base-2 logarithm.}, 
implying an $O(\log n)$ upper bound on $\PoND$. 
We then prove that the bound cannot be improved by more than an $O(\log\log n)$ factor, even if the success probability function is additive. 
Thus, the $\PoND$ is bounded between $\Omega(\frac{\log n}{\log \log n})$ and $O(\log n)$. 

\medskip

\noindent \textbf{Main Result 1:}  $\PoND=\Tilde{\Theta}(\log n)$.

\medskip

As $n$ increases, the bound $\Theta({\log n})$ becomes large, which means that optimality and non-discrimination are hardly compatible.
We observe that if we are willing to relax the non-discrimination requirement, which is often acceptable in practice, the utility can be substantially increased.
To more thoroughly and accurately examine this effect, for $\beta \geq 1$, we extend the notion of non-discrimination to $\beta$-approximate non-discriminatory contracts (or $\beta$-non-discriminatory contracts), which permit the difference between the highest and lowest payments to be within a factor of $\beta$. 
When $\beta = 1$, all agents receive equal payments; as $\beta$ approaches infinity, agents are paid according to the optimal unconstrained contract. 
Accordingly, the price of $\beta$-non-discrimination ($\PoND(\beta)$) can be defined in the same manner.

We then provide a comprehensive parametric study on $\PoND(\beta)$, by setting $\beta=n^{\delta}$, where $\delta\ge 0$. When $\delta=0$, the problem degenerates to the non-relaxed version and in the following, we assume $\delta> 0$.
Our goal is to characterize the relationship between $\PoND(\beta)$ and $\delta$.

\medskip

\noindent \textbf{Main Result 2:} 
Letting $\beta=n^{\delta}$,
\begin{itemize}
    \item for $0<\delta<1$, $\max\{2, \frac{1}{\delta+o(1)}\}\le \PoND(\beta) \le \lceil \frac{1}{\delta} \rceil+1$;
    \item for $\delta=1$, $1.591\le \PoND(\beta) \le 2$;
    \item for $\delta>1$, $1\le \PoND(\beta) \le 1+o(1)$.
\end{itemize}

\medskip

The above result offers an algorithmic framework that assists decision makers in selecting an appropriate level of non-discrimination, taking into account the potential loss in the principal's utility associated with non-discriminatory contracts.
The bounds are asymptotically tight for both small and large values of $\delta$. For example, if $\delta$ is set to $\frac{1}{2}$, then the ratio between the highest and lowest payments is within $\sqrt{n}$, and our contract ensures the principal's utility to be at least $\frac{1}{3}$ of the optimal utility. If $\delta=\frac{1}{\log n}$, the ratio of highest and lowest payments is within $2$, and our contract ensures the principal's utility to be at least $\frac{1}{\log n + 1}$ of the optimal utility. 
Besides, the bound is not tight for values of $\delta$ close to 1, which remains an open question for future research.



\medskip

In the previous two sets of results, we focused on the asymptotic bounds as $n$ becomes large.  
To complement this analysis, we present a case study for two agents, where we derive the exact tight bound on $\PoND(\beta)$ using a more algebraic approach.

\medskip

\noindent \textbf{Main Result 3:} When $n=2$, $\PoND(\beta)=1+\frac{1}{\sqrt{\beta+1}}$.  

\medskip

It is worth noting that non-discriminatory contracts in the same multi-agent model are studied in a simultaneous and independent work \cite{castiglioni2025fairteamcontracts}. However, their main focus is on quantifying the difference between non-discriminatory contracts and other forms of \textit{fair} contracts, rather than examining how non-discrimination impacts the optimal utility of the principal. 


\subsection{Related Work}

\paragraph{Algorithmic contract design.} Following two aforementioned papers in microeconomics by \cite{holmstrom1979moral} and \cite{grossman1992analysis}, contract theory has been well explored from an algorithmic perspective. A core problem within is combinatorial contracts, which can be categorized into single-agent multi-action contracts and multi-contracts. \cite{dutting2025combinatorial} studied the former one where an agent chooses any subset of possible actions, bringing dependencies among agent's actions. They showed an optimal contract can be computed efficiently if the success probability function is gross substitutes, while computing the optimal contract is NP-hard for submodular functions. 
As a consequence, a line of works by \cite{ezra2023approximability,dutting2024combinatorial,deo2024supermodular,feldman2025ultra} also extended the model to a combinatorial setting, with the reward function ranging from gross substitutes and submodular to supermodular. They addressed problems in both tractability and hardness.

\paragraph{Multi-agent contracts.} This paper is more closely related to the multi-agent setting pioneered by \cite{babaioff2006combinatorial}. In their setting, a principal incentivizes a team of agents to jointly work on a single project. Each of the $n$ agents chooses to exert effort or not, and optimizing the principal's utility was defined as the combinatorial agency problem. The landscape was significantly expanded by \cite{dutting2023multi}. In their work, more complicated reward functions were considered, in which the payment to each agent in an incentivized set depends on their marginal contribution to the project. They gave constant-factor approximation algorithms for submodular and XOS reward functions, while there is an $\Omega(\sqrt{n})$ impossibility for subadditive functions with $n$ agents. The most recent work by \cite{ezra2023approximability} and \cite{dutting2025multi}showed that no PTAS exists for Submodular functions, and no approximation better than $\Omega(n^{\frac{1}{6}})$ exists. 

Another multi-agent model was studied by \cite{castiglioni2023multi} and \cite{goel2025multiagentcontractdesignbudget}. In their setting, each agent completes their own task with their payments conditioned on individual performance. Thus, it is natural that payments to each agent are linked to their own outcomes. Our work, however, is based on the first multi-agent model where a team of agents work on a single project. It is mainly because in practical, organizations like worker cooperatives make everyone work towards a common goal instead of finishing individual tasks. As there is no standard way to judging someone's outcome, equalizing the payments matters.

\paragraph{Fairness in contracts.} 
Notably, the only prior study of non-discriminatory contracts was conducted by \cite{feng2024price}. The key difference is that they considered the second multi-agent model as previously mentioned, where agents complete individual tasks instead of a single task. While they have derived logarithmic price of non-discrimination across all their settings, our model introduces the concept of price of $\beta$-non-discrimination to achieve constant-factor approximations. 
Fairness concerns were also investigated by \cite{castiglioni2025faircontracts}, who considered the trade-off between profit-maximizing and fairness when allocating a bunch of tasks to some agents. 
In a simultaneous work with this paper, \cite{castiglioni2025fairteamcontracts} explore \textit{fair} contracts using a swap-based definition. Although non-discriminatory contracts are involved, their primary objective is to bound the gap between non-discriminatory contract and optimal \textit{fair} contract.


\section{Preliminaries}
\paragraph{Multi-agent hidden-action setting.} For any $k\in \mathbb{N}^+$, let $[k]=\{1,\ldots,k\}$. We study the model in which a principal interacts with a set of agents $N=[n]=\{1,\ldots,n\}$ to execute a project. Each agent chooses whether to exert effort or shirk. For each agent $i$, exerting effort incurs cost $c_i \in \mathbb{R}_{> 0}$, while shirking incurs zero cost.
The outcome of the project is either ``success'' or ``failure'', and depends on the set of agents who choose to exert effort. 
The principal observes the outcome of the project but not the actions of agents. 
Let $f: 2^{N}\rightarrow [0,1]$ be the project's success probability function. When a set of agents $S$ exert effort, the project succeeds with probability $f(S)$.
If the project succeeds, the principal derives reward normalized to $r=1$. 
Thus, $f$ is also referred to as the (expected) reward function.
Otherwise if the project fails, the reward is zero.
For any $i\in N$ and $S\subseteq N$, define $f(i \, | \, S):= f(S\cup \{i\})-f(S)$ as the marginal contribution of including $i$ to $S$.
In this paper, $f$ is considered to be monotone and submodular, where monotonicity ensures that for any $S\subseteq T$, $f(S)\leq f(T)$, and submodularity ensures that for any $S\subseteq T$ and any $i\notin T$, $f(i \, |\, T)\leq f(i \, |\, S)$.
For any $S\subseteq N$, $|S|$ denotes its cardinality.

\paragraph{Optimal contract design.} Since exerting effort is costly, agents by themselves have no interest in exerting effort, of which the benefit goes to the principal.
To incentivize agents, the principal designs a contract that assigns payments to agents based on the outcome of the project. 
In this paper, we focus on the \emph{linear contract}, an $n$-dimensional vector $\alpha=(\alpha_1,\ldots,\alpha_n)$ with every $\alpha_i\in[0,1]$. A contract $\alpha$ assigns payment $\alpha_i$ to agent $i$ when the project succeeds, regardless of whether agent $i$ exerts effort or not.

When analyzing contracts, we consider the pure Nash equilibria of the induced game among the agents and adopt the tie-breaking rule where agents exert effort when indifferent.
For a contract $\alpha$, suppose that it incentivizes agents $S$ to exert effort, i.e.,
\[
\begin{aligned}
    &\alpha_if(S)-c_i\geq \alpha_if(S\setminus\{i\}) \quad \text{for all } i\in S, \\
    &\alpha_if(S) \geq \alpha_if(S\cup\{i\}) - c_i \quad \text{for all } i\notin S.
\end{aligned}
\]
Note that each agent $i\in S$ has expected utility $\alpha_if(S)-c_i$, while each $i\notin S$ has expected utility $\alpha_if(S)$. 
Then for contract $\alpha$ and the equilibrium where agents in $S$ exert effort, the principal's \emph{utility} is 
\[
(1-\sum_{i=1}^n\alpha_i) f(S).
\]
The principal pursues the contract that maximizes her utility. Then for a fixed set of agents $S$, the contract $\alpha$ that maximizes the principal's utility among all contracts in which the set of agents exerting effort is exactly $S$, is
\[
\begin{aligned}
    &\alpha_i=\frac{c_i}{f(i \, | \, S\setminus \{i\})} \quad &\text{for all } i\in S, \\
    &\alpha_i=0 \quad & \text{for all } i\notin S.
\end{aligned}
\]
For any $S\subseteq N$, define
\[
g(S):=\left( 1-\sum_{i\in S} \frac{c_i}{f(i \,| \, S\setminus \{i\})} \right) f(S),
\]
the principal's optimal utility for the contracts where the set of agents exerting effort is $S$.
Therefore, the problem of finding the contract with the maximum utility of the principal is equivalent to $\max_{S\subseteq 2^N} g(S)$, which, however, is known to be NP-hard \cite{dutting2023multi}.

\paragraph{Approximately non-discriminatory contracts.}
We now introduce non-discriminatory (linear) contracts. 
Informally, a contract is non-discriminatory (ND) if the payments assigned to agents who exert effort are identical.
The ND requirement does not apply to the agents who choose to shirk.
For a fixed set of agents $S$, the ND contract $\alpha$ that incentivizes $S$ and maximizes the principal's utility satisfies
\[
\begin{aligned}
    &\alpha_i=\max\limits_{j\in S}\frac{c_j}{f(j \, | \, S\setminus \{j\})} \quad &\text{for all } i\in S, \\
    &\alpha_i=0 \quad & \text{for all } i\notin S.
\end{aligned}
\]
For any $S\subseteq N$, define
\[
\alpha_S:=\max_{j\in S}\frac{c_j}{f(j \, | \, S\setminus \{j\})}.
\]
The contract that pays $\alpha_S$ to each agent in $S$ and zero to all other agents incentivizes exactly the agents in $S$ and attains the minimum total payment among all contracts that induce effort by $S$.
We also define
\[
g_{ND}(S):=\left( 1-|S|\alpha_S \right) f(S).
\]
The interpretation of $g_{ND}(S)$ is the principal's optimal utility under the ND contract that induces effort by $S$.
Therefore, in this context, the principal faces the optimization problem of $\max_{S\subseteq 2^N}g_{ND}(S)$. 
Throughout the paper, we refer to contracts without the ND constraint as \emph{unconstrained} contracts.
When incentivizing the same set of agents $S$, the payments for agents under the optimal ND contract are weakly higher than those under the unconstrained contract. 
Hence, the non-discrimination requirement reduces the principal’s utility.

In addition to the exact non-discrimination, we are also concerned with approximately non-discriminatory contracts, in which payments assigned to agents who exert effort do not differ substantially.
To measure the disparity among agents' payments, we use a multiplicative factor $\beta$ with $\beta\geq 1$, standing for the wage ratio.
For a \emph{$\beta$-non-discriminatory} ($\beta$-ND) contract $\alpha$, if it incentives $S$ to exert effort, then for any $i,j\in S$ with $\alpha_i\le \alpha_j$, $\frac{\alpha_j}{\beta} \le \alpha_i\leq \alpha_j$ holds.
For a fixed set of agents $S$, the $\beta$-ND contract $\alpha$ that incentivizes $S$ and maximizes the principal's utility satisfies
\[
\begin{aligned}
    &\alpha_i= \max\left\{ \frac{c_i}{f(i \,| \,S\setminus\{i\})},\frac{1}{\beta} \alpha_S
    \right\} \quad &\text{for all } i\in S, \\
    &\alpha_i=0 \quad & \text{for all } i\notin S.
\end{aligned}
\]
For any $S\subseteq N$ and $\beta\geq 1$, define $g_{ND}(\beta,S)$ as
\[
\left( 1- \sum_{i \in S} \max\left\{ \frac{c_i}{f(i \,| \,S\setminus\{i\})},\frac{1}{\beta} \alpha_S \right\} \right)f(S),
\]
which stands for the principal's utility under the optimal $\beta$-ND contract in which the set of agents exerting effort is $S$.
Note that $g_{ND}(1,S)=g_{ND}(S)$. For $\beta$-ND contracts, the optimization problem for principal is $\max_{S\subseteq 2^N}g_{ND}(\beta,S)$.


\paragraph{The price of non-discrimination.}

We quantify the optimal utility loss due to the requirement of non-discrimination through the framework of the price of $\beta$-non-discrimination, a function with respect to $\beta$. 
Informally, it is the supremum ratio between the principal's optimal utility of all contracts and the principal's optimal utility of all $\beta$-ND contracts.
Denote by $\mathcal{I}$ the set of problem instances. For any instance
$I\in \mathcal{I}$, let $\OPT(I)$ be the principal's optimal utility of all contracts and let $\OPT_{ND}(\beta,I)$ be the principal's optimal utility of all $\beta$-ND contracts.

\begin{definition}
    For any $\beta\geq 1$, the price of $\beta$-non-discrimination is
    \[
    \PoND(\beta) := \sup\limits_{I\in \mathcal{I}} \frac{\OPT(I)}{\OPT_{ND}(\beta,I)}.
    \]
\end{definition}
When $I$ and $\beta$ are clear from the context, we use $\OPT$ and $\OPT_{ND}$ to refer to $\OPT(I)$ and $\OPT_{ND}(\beta,I)$, respectively. Also, $\PoND$ denotes $\PoND(1)$ for simplicity.

\section{Logarithmic PoND}\label{sec:nd}
In this section, we study the price of (exact) non-discrimination. 
We show that there always exists a non-discriminatory contract achieving at least $\frac{1}{\log n}$ of the unconstrained optimum, and this is almost tight, as there exist instances where the ratio is at most $O(\frac{\log \log n}{\log n})$.

\begin{theorem}\label{nd}
    For $n$ agents with submodular reward function $f$, the price of 1-non-discrimination is at least $\Omega(\frac{\log n}{\log \log n})$ and at most $O(\log n)$.
\end{theorem}

The difficulty for deriving the $O(\log n)$ upper bound lies in identifying the set of agents to be incentivized by the contract.
A natural greedy approach that adds agents sequentially, for example in decreasing order of ``price to performance ratio\footnote{For an agent, his price-performance ratio is the ratio between his marginal contribution and his payment in the optimal contract.}'', is unlikely to work.
The rationale is that including an agent with huge payment gap can dramatically reduce the principal’s utility, since under non-discrimination all incentivized agents are equally paid. 
Skipping that high‑reward agent does not necessarily resolve the issue, because the current set of incentivized agents may already be at a threshold where adding any additional agent pushes the total payment above 1.

Instead, we design the contract based on a particular configurations of agents. 
Begin with the optimal unconstrained contract in which the agents in $S^*$ are incentivized to exert effort. 
We partition $S^*$ into several groups $G_1,\ldots,G_m$, such that (1) the total payment to each group under ND contracts does not exceed the total payment to $S^*$ in the optimal unconstrained contract, 
and (2) the sum of the principal’s reward over the groups under ND contract is at least the reward under the optimal unconstrained contract.


\begin{lemma} \label{nd_ub}
    Suppose $S^*$ is the set of agents exerting effort in the optimal unconstrained contract.
    When $f$ is submodular, there exists a subset of agents $S\subseteq S^*$ such that 
    \[
    g_{ND}(S)=f(S)(1-|S| \alpha_S)\ge \frac{1}{\lceil\log n\rceil} g(S^*).
    \]
\end{lemma}
\begin{proof}
    Since $|S^*|\leq n$, it suffices to show $g_{ND}(S)\geq \frac{1}{\lceil \log |S^*| \rceil}g(S^*)$.
    Assume without loss of generality that the optimal unconstrained contract incentivizes all $n$ agents, i.e. $S^*=N$.
    For each $i\in N$, define $\alpha_{i,N}=\frac{c_i}{f(i \, | \, N \setminus \{i\})}$. Reindex the agents according to a descending order of $\alpha_{i,N}$, such that $\alpha_{1,N}\ge ...\ge \alpha_{n,N}$.

    We now partition $N$ into $m=\lceil\log n \rceil$ groups $G_1,\ldots,G_m$ as follows: for each $k\in [m-1]$, group $G_k$ contains $2^{k-1}$ consecutive agents $2^{k-1},2^{k-1}+1,\ldots,2^k-1$; the remaining agents form group $G_m$.
    Informally, group $G_m$ contains the half of $N$ with smaller $\alpha_{i,N}$, and group $G_{m-1}$ contains the half of $N\setminus G_m$ with smaller $\alpha_{i,N}$, and so on.
    We illustrate this partition in Figure \ref{fig:lem2}.

\begin{figure*}[t]
    \centering
    \includegraphics[width=1\textwidth]{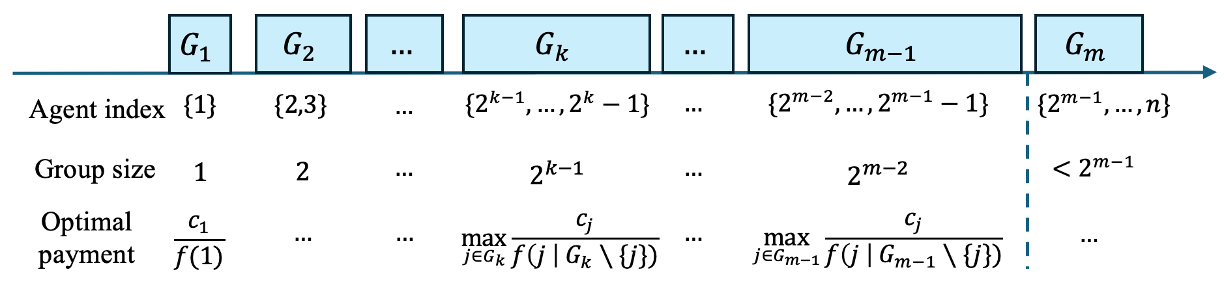}
    \caption{Our partition of group $G_1,...,G_k$ in Lemma \ref{nd_ub}. Each group is double the size of the previous one, except for the last group. The optimal payment of group $k$ stands for the standardized payment when incentivizing that single group only.}
    \label{fig:lem2}
\end{figure*}

    Consider incentivizing all agents in a single group $G_k$ under ND contracts.
    The optimal ND contract that incentivizes exactly $G_k$ pays each $i \in G_k$ the amount
    $\alpha_{G_k}=\max_{j\in G_k}{\frac{c_j}{f(j \,|\,G_k \setminus \{j\}) }} $ and pays zero to agents in $N\setminus G_k$.
    The total payment to agents in $G_k$ equals
    \begin{align}\label{eq::log-1}
    \begin{split}
        &|G_k|\cdot \alpha_{G_k} =|G_k|\cdot\max_{j\in G_k}{\frac{c_j}{f(j \,|\,G_k \setminus \{j\}) }}\\
        \overset{(a)}{\le} &|G_k|\cdot\max_{j\in G_k}{\frac{c_j}{f(j \,|\,N \setminus \{j\}) }} \\
        \overset{(b)}{\le}&\sum_{i=1}^{|G_k|-1}{\frac{c_i}{f(i \,|\,N \setminus \{i\})}}+\max_{j\in G_k}{\frac{c_j}{f(j \,|\,N \setminus \{j\}) }}\\
        =&\sum_{i=1}^{|G_k|}{\alpha_{i,N}}\le \sum_{i=1}^{n}{\alpha_{i,N}},
    \end{split}
    \end{align}
    where $(a)$ holds due to $G_k\subseteq N$ and the submodularity of $f$ and $(b)$ is due to that $\alpha_{i,N}$'s are sorted in descending order and $|G_1|+|G_2|+\cdots+|G_{k-1}| \ge |G_k|-1$.

    Inequality (\ref{eq::log-1}) implies that for each group $G_k$, the principal's maximum utility when incentivizing $G_k$ is 
    \begin{align*}
        \begin{split}
            g_{ND}(G_k)&=f(G_k)(1-|G_k|\cdot \alpha_{G_k})\ge f(G_k)(1-\sum_{i=1}^{n}{\alpha_{i,N}}).
        \end{split}
    \end{align*}
    Summing over $k\in [m]$ yields
    \begin{align*}
        \begin{split}
        \sum_{k=1}^{m}g_{ND}(G_k)&\ge \left(\sum_{k=1}^{m}f(G_k)\right)(1-\sum_{i=1}^{n}{\alpha_{i,N}})\\
        &\ge f(N)(1-\sum_{i=1}^{n}{\alpha_{i,N}})=\OPT,
        \end{split}
    \end{align*}
    where the second inequality transition is due to the submodularity of $f$. 
    The above inequality implies that the maximum of $g_{ND}(G_k)$ among all $k\in [m]$ is at least $\frac{1}{m}\OPT=\frac{1}{\lceil\log n\rceil}\OPT$. 
\end{proof}

Lemma \ref{nd_ub} indicates $\PoND \le O(\log n)$.
For the lower bound, we construct an additive instance with about $\log n$ groups of agents, homogeneous within each group.
Incentivizing any group yields the same reward.
As shown in this example, the optimal unconstrained contract incentivizes all groups, while any ND contract that yields positive utility for the principal can incentivize at most $O(\frac{\log \log n}{\log n})$ groups.

\begin{lemma}\label{nd_lb}
    There exists an instance with additive $f$ such that the principal's utility under any ND contract is no greater than a $O(\frac{\log \log n}{\log n})$ fraction of that under the optimal unconstrained contract. 
\end{lemma}

\begin{proof}
    Let $n=2^m-1$ where $m$ is an integer.
    In the instance, the $n$ agents are classified into $m$ groups $G_1, G_2,\ldots,G_m$, where for any $k\in [m]$, group $G_k$ contains $2^{k-1}$ agents.
    Agents within a group share the same cost function reward function.
    In particular, for any $G_k$ and any $i\in G_k$,  $f(\{i\})=\frac{1}{m2^{k-1}}$ and $c_i=\frac{1}{Tm^24^{k-1}}$, where $T\geq 2$ is a constant number.
    Thus, to incentivize agent $i\in G_k$, the principal must pay at least $\frac{1}{Tm2^{k-1}}$.

    We claim that the optimal utility of the principal under the unconstrained contract is achieved by incentivizing all $n$ agents. 
    If $A$ is a strict subset of $N$,
    then there exists an agent $j \in N$ with $j\notin A$.
    Define $A':=A\cup \{j\}$.
    One can verify that the difference between the utilities at $A'$ and $A$ equals
    \[
    f(\{j\})\left(1-\sum_{i\in A } \alpha_i - \alpha_j\right) - f(A)\alpha_j.
    \]
    Observe for any $i \in N$, $\alpha_i=\frac{f(\{i\})}{T}$. Since $\sum_{i\in A}\alpha_i \leq \frac{1}{T}$ and $f(A)\leq 1$, the difference is at least
    $
    f(\{j\})\frac{T-1}{T}-\alpha_j= f(\{j\})\frac{T-2}{T}>0.
    $
    Therefore, $g(S)$ is maximized at $S=N$ with value $1-\frac{1}{T}$.

Next, we turn to ND contracts and suppose that the optimal ND contract incentivizes only a set $S \subseteq N$. 
Then in this ND contract, each agent in $N\setminus S$ receives zero payment.
Define $L:=\min\{k:G_k\cap S \neq \emptyset\}$.
Since each agent in $S$ receives the same payment $\frac{1}{Tm2^{L-1}}$ under the optimal ND contract, and agents in the group with a lower index have higher marginal contributions to $f$, the principal must prioritize agents from lower-index groups. 
Hence, if $S$ contains some agent(s) in $G_k$ for $k\geq L+1$, then all agents from $G_L$ to $G_{k-1}$ must belong to $S$.

Assume without loss of generality that $f(S)=\lambda\le 1$, since $f(N)=1$. 
Because every group contributes $\frac{1}{m}$ to the total reward, $S$ must contain at least $t=\floor*{\lambda m}$ entire groups, that is, $G_L\cup G_{L+1}\cup \cdots \cup G_{L+t-1} \subseteq S$.
Then the number of agents in $S$ satisfies
\[
|S|\geq 2^{L-1}\sum_{j=1}^{t}{2^{j-1}}=2^{L-1}(2^{\floor*{\lambda  m}}-1).
\]
Since each agent in $S$ receives payment $\frac{1}{Tm 2^{L-1}}$, the principal's optimal utility $g_{ND}(S)$ is at most
\begin{align*}
    \begin{split}
        \lambda \Big( 1- \frac{2^{L-1}(2^{\floor*{\lambda m}}-1)}{Tm 2^{L-1}}\Big)
        =\lambda \Big( 1-\frac{2^{\floor*{\lambda m}}-1}{Tm }\Big).
    \end{split}
\end{align*}
Since an ND contract that incentivizes a single group yields positive utility for the principal, the upper bound on $g_{ND}(S)$ must be positive.
Then $2^{\floor*{\lambda  m}}-1\le Tm$ holds, and thus, $\lambda\leq \frac{\log (Tm+1)}{m}$. 
Therefore, it holds that
\begin{align*}
    \begin{split}
         \frac{\OPT}{\OPT_{ND}} > \frac{T-1}{T}&\frac{m}{\log (Tm+1)},
    \end{split}
\end{align*}
implying that $\PoND = \Omega(\frac{\log n}{\log\log n})$ as $n=2^m-1$.
\end{proof}



Theorem \ref{nd} directly follows from Lemmas \ref{nd_ub} and \ref{nd_lb}. 
We conjecture that the price of non-discrimination is tight on the current lower bound.

\begin{conjecture}
    For $n$ agents with submodular valuation functions, $\PoND=\Theta (\frac{\log n}{\log \log n})$.
\end{conjecture}

\section{Constant-Factor PoND(\texorpdfstring{$\beta$}{beta})}\label{sec:bnd}

In the bad example, a non-discriminatory contract yields about $O(\frac{1}{\log n})$ of the optimal utility, driven by large disparities in unconstrained payments across agents. A natural question raised is how much the price of non-discrimination improves if the ratio of the maximum to minimum payment among incentivized agents is bounded by a factor $\beta$.

In this section, we show that when $\beta=n^\delta$ for some constant $\delta \in (0,1]$
, $\PoND(\beta)$ is bounded by a constant. The case when $\delta>1$ follows from the remark below:


\begin{remark}\label{r1}
    For $\delta>1$, one can achieve exact-approximation with sufficiently large $n$: For agents $N$ incentivized under the optimal contract, the total payment with and without ND differs by at most $\frac{(n-1)\alpha_{i}}{n}\rightarrow0$, where $\alpha_{i}$ is the maximum payment value among $N$. 
    Thus, we restrict our attention to $\delta \in (0,1]$. 
\end{remark}

\begin{theorem}\label{bnd}
    Suppose $f$ is submodular. 
    For any $\delta \in (0,1]$ and sufficiently large $n$, the price of $n^\delta$-non-discrimination is at most $\lceil \frac{1}{\delta} \rceil+1$ and at least $\max\{2, \frac{1}{\delta+O(\frac{\log\log n}{\log n})}\}$ when $\delta \in (0,1)$ and at least $1.591$ when $\delta=1$.
\end{theorem}

The lemma below establishes the upper bound on the price of $n^\delta$-non-discrimination. 


\begin{lemma}\label{bnd_ub}
    Suppose that $f$ is submodular and $S^*$ is the set of agents exerting effort in the optimal unconstrained contract. 
    For any $\delta \in (0,1]$ and sufficiently large $n$, there exists an ND contract in which agents $S\subseteq S^*$ exert effort and
    \[
    g_{ND}(n^{\delta},S)\ge \frac{1}{\lceil\frac{1}{\delta}\rceil+1} g(S^*).
    \]
\end{lemma}

\begin{proof}
    As our objective is to design an ND contract that incentivizes a subset of agents in $S^*$, assume without loss of generality that the optimal unconstrained contract incentivizes all $n$ agents as we did in Lemma \ref{nd_ub}, i.e. $S^*=N$.
    For any $i\in N$, define $\alpha_{i,N}:=\frac{c_i}{f(i \,|\, N \setminus \{i\})}$.
    By the optimality of $g(S^*)$, we have $\alpha_{i,N} < 1$ for all $i$. We now group agents by their $\alpha_{i,N}$ values.

    Let $t=\lceil \frac{1}{\delta} \rceil$ and we partition agents into $t+1$ groups $G_1,G_2,\ldots,G_{t+1}$, defined as follows:
    \[
    \begin{aligned}
        & G_1=\left\{i\in N : 0\leq \alpha_{i,N}<\frac{1}{n} \right\}, \\
        &G_2=\left\{i\in N : \frac{1}{n}\leq \alpha_{i,N}< \frac{1}{n^{(t-1)\delta}} \right\},
    \end{aligned}
    \]
    and for any $j=3,\ldots,t+1$,
    \[
    G_j=\left\{i\in N: \frac{1}{n^{(t+2-j)\delta}} \leq \alpha_{i,N}< \frac{1}{n^{(t+1-j)\delta}} \right\}.
    \]
    Note that group $G_2$ is well defined because $(t-1)\delta < 1$. Starting from the right boundary of $G_2$, repeatedly extend to the right by a factor of $\frac{1}{\delta}$ to obtain groups $G_3,\ldots,G_{t+1}$.

    We then consider $\beta$-ND contracts that incentivize each group $G_k$ separately 
    and show that the sum of the principal’s utilities from these contracts is at least the optimal utility under the unconstrained contract.
    For any agent $i$, suppose that $i\in G_k$. Define
    \[
    \hat{\alpha}_i:=\max\left\{ \frac{c_i}{f(i \,|\, G_k \setminus \{i\})}, \frac{1}{n^\delta}\max_{j \in G_k}{\frac{c_j}{f(j \,|\, G_k \setminus \{j\})}}  \right\}
    \]
    the minimum payment to $i$ under the $n^\delta$-ND contract that induces effort by exactly $G_k$.
    
    First, consider group $G_k$ with $k\geq 2$. 
    Since $f$ is submodular, then $\sum_{i\in G_k} \hat{\alpha}_i$ is at most
    \[
    \sum_{i\in G_k}
        \max\bigg\{
        \frac{c_i}{f(i \mid N \setminus \{i\})},\frac{1}{n^\delta}\max_{j \in G_k}
        \frac{c_j}{f(j \mid N \setminus \{j\})}
        \bigg\},
    \]
    which equals $\sum_{i\in G_k}\alpha_{i,N}$ as $\alpha_{i,N}\geq \frac{1}{n^\delta}\max_{j\in G_k} \alpha_{j,N}$.

    For group $G_1$, since $\hat{\alpha}_i\leq \alpha_{i,N} < \frac{1}{n}$ for all $i\in G_1$, we have
    \[
    \sum_{i\in G_1} \hat{\alpha}_i < \frac{|G_1|}{n^{(1+\delta)}} + \sum_{i\in G_1} \alpha_{i,N}.
    \]

    Since $|G_1|\le n$ and $n$ is sufficiently large, the above inequality implies $\sum_{i\in G_1} \hat{\alpha}_i < \sum_{i\in G_1} \alpha_{i,N}+o(1)$.
    Note that setting payment $\hat{\alpha}_i$ to every $i\in G_k$ and zero to others yields the optimal $n^\delta$-ND contract that induces effort exactly by $G_k$.
    Under this contract, the principal has utility
    \[
    f(G_k)\left( 1- \sum_{i \in G_k} \hat{\alpha}_i \right) = g_{ND}(n^\delta,G_k).
    \]
    We design such a contract for each group $G_k$. 
    Then the principal’s total utility from $t+1$ separate contracts is at least
    \[
    \begin{aligned}
          &\sum_{k=1}^{t+1}f(G_k)( 1-\sum_{i\in G_k}\alpha_{i,N} ) - f(G_1)\cdot o(1) \\
         \geq &f(G_1)( 1-\sum_{i\in N} \alpha_{i,N}) + \sum_{k=2}^{t+1}f(G_k)( 1-\sum_{i\in N} \alpha_{i,N}) \\
           \geq &f(N)( 1-\sum_{i\in N} \alpha_{i,N}) = \OPT,
    \end{aligned}
    \]
    where the first inequality transition is due to $G_k\subseteq N$ for all $k$ and the second inequality uses submodularity of $f$.
    By the above inequality, there exists a group $G_{k^*}$ such that the $n^\delta$-ND contract with payment $\hat{\alpha}_i$ to every $i\in G_{k^*}$ (and $0$ to all other agents) yields utility for the principal of at least $\frac{1}{t+1}\OPT=\frac{1}{\lceil \frac{1}{\delta} \rceil+1}\OPT$.
\end{proof}

We complement the upper bound with a counter example, in which the optimal $n^\delta$-ND contract achieves at most $\delta+o(1)$ fraction of the unconstrained optimum.


\begin{lemma}\label{bnd_epslb}
    For any $\delta \in (0,1)$, there exists an instance where the principal's optimal utility under any $n^\delta$-ND contract is at most a $(\delta+O(\frac{\log\log n}{\log n}))$ fraction of the optimal utility under unconstrained contracts. 
\end{lemma}

\begin{figure}[t] 
\centering 
\includegraphics[width=0.7\linewidth]{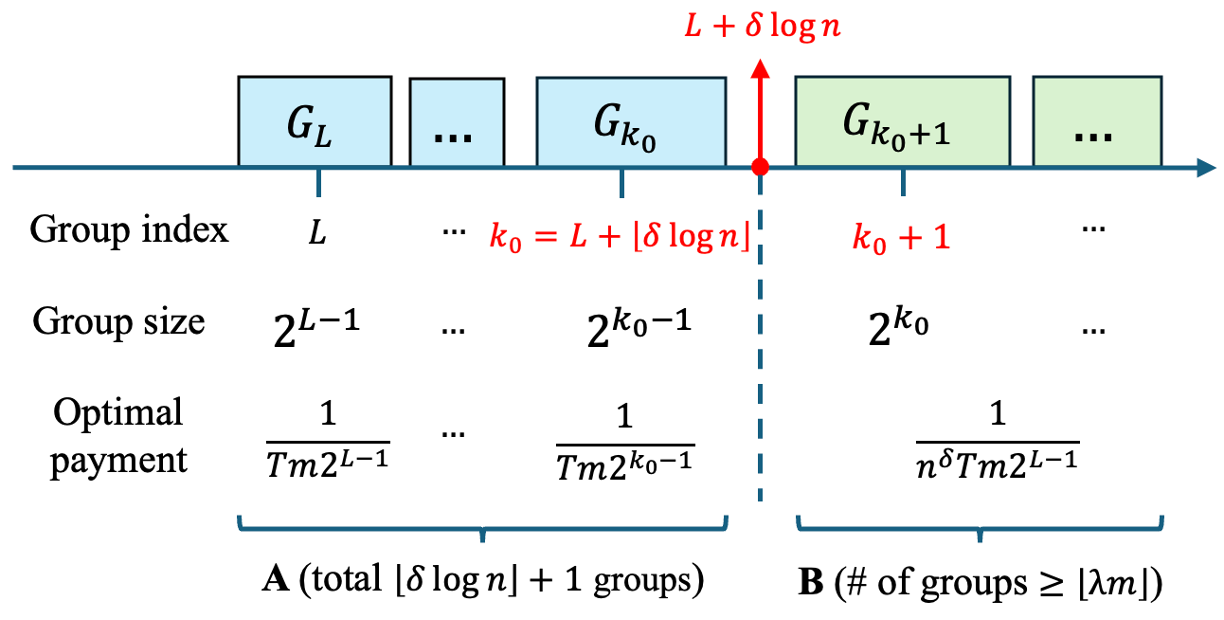} 
\caption{Illustration of our example in Lemma \ref{bnd_epslb}. The blue part ($A$) gets the same payment as in the optimal unconstrained contract, while agents in green part ($B$) receive a uniform payment, which equals the highest payment to agents in $A$ divided by ${n^\delta}$.} \label{fig:lem7} 
\end{figure}

We outline the proof of Lemma \ref{bnd_epslb} and defer its proof to Appendix \ref{a4}. To sketch, we consider the very same instance as that in Lemma \ref{nd_lb}, where the agents belonging to lower-indexed groups have higher success probability and payment threshold, and vice versa. Based on our definition of $\beta$-ND, we partition all agents into $A$ and $B$ according to a ``pivot" group $k_0$, as illustrated in Figure \ref{fig:lem7}. Agents in groups with index smaller than or equal to $k_0$ are paid the original amount as in the optimal unconstrained contract, while agents in groups with higher index are paid a uniform amount. We upper-bound $\OPT_{ND}$ by incentivizing as many agents from $A$, plus some consecutive groups from $B$.

We also derive another lower bound that improves the bound in Lemma~\ref{bnd_epslb} in some cases, e.g., when $\delta \geq \frac{1}{2}$.

\begin{lemma}\label{bnd_2lb}
    For any $\delta \in (0,1)$ and arbitrarily small $\epsilon>0$, there exists an instance where the principal's optimal utility under $n^\delta$-ND contracts is at most a $(\frac{1}{2}+\epsilon)$ fraction of that under unconstrained contracts.

\end{lemma}

The details of the construction are given in Appendix \ref{a4}. The idea is as follows: the upper bound shown in Lemma \ref{bnd_ub} requires payment disparities. 
We exploit such disparities by creating an agent with high contribution and payment ($A$) and a group of ``free labors" with very low contribution and payment ($B$). 
In order to incentivize the powerful agent in $A$, the principal is forced to exclude all the agents in $B$.

To conclude, we provide the lower bound for $\delta$ = 1. 
\begin{lemma}\label{bnd_1.6lb}
    There exists an instance where the principal's optimal utility under $n$-ND contracts is no greater than $0.629$ fraction of that under unconstrained contracts.
\end{lemma}

Basically, the construction is analogous to that of Lemma \ref{bnd_2lb}, but we can additionally incentivize a portion of agents in $B$. We defer the discussion to Appendix \ref{a4}.

\section{Tight Results for Two Agents}\label{sec:2a}
Our results in the previous section are grounded in large $n$. 
In this section, we shift our focus to the case of two agents, in which we are able to identify the optimal $\beta$-ND contract, and establish tight $\PoND(\beta)$ for all $\beta\ge 1$.

\begin{theorem}\label{bnd:2a}
    For any $\beta \geq 1$, when there are two agents and the reward function $f$ is submodular, the $\PoND(\beta)$ equals $(1+\frac{1}{\sqrt{\beta+1}})$, and this bound is tight.
\end{theorem}
We defer the discussion of this theorem to Appendix \ref{a5}. By enumerating all the ways to incentivize agent(s), 
the $\PoND(\beta)$ is characterized by algebraically proving an inequality. 

\section{Conclusion}

This paper investigates how varying degrees of non-discrimination influence the optimal revenue in multi-agent contracts. 
We provide asymptotic bounds on the price of $\beta$-non-discrimination, characterizing the tradeoff between non-discrimination and revenue. 
Our bounds are nearly tight for both the strict and relaxed cases when $\beta$ is either small or large. 
However, determining the exact bound when $\beta$ approaches $n$ remains an open question for future research.

Our results uncover several promising directions for further study. 
For example, it would be valuable to extend our framework to other classes of reward functions, such as fractionally subadditive (XOS) and subadditive functions. 
It is also interesting to investigate alternative objectives, such as social welfare, social surplus, and various fairness criteria.



\newpage

\bibliographystyle{plain}
\bibliography{ref}

\newpage
\appendix

\section*{Appendix}

\section{Missing Proofs from Section 4}\label{a4}

\noindent\textbf{Lemma 7.} \textit{For any $\delta \in (0,1)$, there exists an instance where the principal's optimal utility under any $n^\delta$-ND contract is at most $(\delta+O(\frac{\log\log n}{\log n}))$ fraction of the optimal utility under unconstrained contracts.} 
\begin{proof}
    We consider the same instance as in the proof of Lemma~\ref{nd_lb}.
    Recall that $n=2^m-1$ agents are partitioned into $m$ groups $G_1,\ldots,G_m$.
    For any $G_k$ and $i\in G_k$, $c_i=\frac{1}{Tm^24^{k-1}}$ and $f(\{i\})=\frac{1}{m2^{k-1}}$, where $T\ge 3$ is a constant.
    To incentivize agent $i$, payment at least $\alpha_i=\frac{1}{Tm2^{k-1}}$ is required.
    For each $G_k$, $f(G_k)=\frac{1}{m}$. The optimal unconstrained contract incentivizes all agents and achieves $\OPT = 1-\frac{1}{T}$.

    We now consider $n^{\delta}$-ND contracts and suppose that the optimal one induces effort by a set $S$ of agents.
    Define $L:=\min\{k:G_k\cap S \neq \emptyset\}$.
    For any $i\in S$ (suppose $i\in G_k$), define
    \[
    \hat{\alpha}_i= \max\left\{ \frac{1}{Tm2^{k-1}},\frac{1}{n^\delta Tm 2^{L-1}}\right\}.
    \]
    The optimal $n^{\delta}$-ND contract pays $\hat{\alpha}_i$ to every $i\in S$ and zero to all other agents.
    Define $k_0:=L+\lfloor\delta \log n \rfloor$, $A:=(G_L\cup \cdots \cup G_{k_0})\cap S$, and $B:=S\setminus A$.
    It is not hard to verify that $\hat{\alpha}_i=\alpha_i$ for all $i\in A$ and $\hat{\alpha}_i = \frac{1}{n^\delta Tm 2^{L-1}}$ for all $i\in B$.
    Figure~\ref{fig:lem7} illustrates the structure of $S$.

    We next present the upper bound of the principal's utility under the optimal $n^\delta$-ND contract.
    Since $f$ is additive,  the optimal $n^\delta$-ND contract gives utility
    \[
    \left(f(A)+f(B) \right)(1-\sum_{i\in A }\hat{\alpha}_i - \sum_{i\in B}\hat{\alpha}_i),
    \]
    which is at most
    \[
    f(A)\left( 1-\sum_{i\in A}\alpha_i\right) + f(B).
    \]

    Viewing $f(A)\left(1-\sum_{i\in A}\alpha_i\right)$ as function of $A$ and using the analogous strategy to Lemma \ref{nd_lb}, the maximum of the function is attained at $A=G_L\cup \cdots \cup G_{k_0}$, and the maximum value is less than 
    \[
    f(G_L\cup\cdots \cup G_{k_0})= \frac{\lfloor \delta \log n \rfloor + 1}{m}.
    \]
    
    We next establish the upper bound of $f(B)$ (given that $B\subseteq S$) using arguments analogous to the proof of Lemma~\ref{nd_lb}.
    In the optimal $\beta$-ND contract, agents in $B$ receive the same payment.
    The principal therefore prioritizes groups with lower indices because they contribute higher success probabilities. As a consequence, $B$ consists of consecutive groups starting from $G_{k_0+1}$, with the last group possibly partial.

    Assume $f(B)=\lambda$, and we will bound $\lambda$ from the above. Since each entire group in $B$ contributes success probability $\frac{1}{m}$, $B$ contains at least $t=\lfloor \lambda m \rfloor$ consecutive entire groups.
    For each $G_k\subseteq S$ (hence $k>k_0$), the total payment $\sum_{i\in G_k} \hat{\alpha}_i$ equals
    \[
    \frac{2^{k-1}}{Tmn^\delta 2^{L-1}} = \frac{2^{k-L}}{Tmn^\delta}\geq \frac{2^{k-k_0-1+\delta\log n}}{Tmn^\delta}=\frac{2^{k-k_0-1}}{Tm},
    \]
    where the inequality transition is due to $L\leq k_0+1-\delta\log n$.
    Summing the inequality over all $t$ groups in $B$ yields
    \begin{align*}
        \begin{split}
            \sum_{i \in B}\hat{\alpha}_i\ge \sum_{k=k_0+1}^{k_0+t}\frac{2^{k-k_0-1}}{Tm}=\frac{2^{t}-1}{Tm}=\frac{2^{\lfloor \lambda m\rfloor}-1}{Tm}.
        \end{split}
    \end{align*}
    Since $B \subseteq S$, we must have $\sum_{i \in B}\hat{\alpha}_i < 1$; otherwise, the $n^{\delta}$-ND contract inducing effort by $S$ yields zero utility, whereas incentivizing a single group yields positive utility.
    Combining the above inequality and $\sum_{i \in B}\hat{\alpha}_i < 1$ gives $\lambda=O(\frac{\log m}{m})=O(\frac{\log \log n}{\log n})$.

    Since $m=\log (n+1)$, the principal's optimal utility under $n^\delta$-ND contracts is at most
    \[
    \frac{\lfloor \delta \log n \rfloor + 1}{m}+\lambda \leq \delta+O(\frac{\log \log n}{\log n}).
    \]
    By letting $T$ become large, the above upper bound with $\OPT=1-\frac{1}{T}$ gives $\PoND(n^\delta)\ge \frac{1}{\delta+O(\frac{\log \log n}{\log n})}$. 
\end{proof}

\noindent\textbf{Lemma 8.} \textit{For any $\delta \in (0,1)$ and arbitrarily small $\epsilon>0$, there exists an instance in which the principal's optimal utility under $n^\delta$-ND contracts is at most a $(\frac{1}{2}+\epsilon)$ fraction of the optimal utility under unconstrained contracts.}
\begin{proof}
    Let $M > 3+\frac{1}{\epsilon}$ and $n \gg M^{\frac{1}{1-\delta}}$.
    Consider an instance with additive $f$ and $n$ agents.
    Agents are classified into two groups, namely $A$ and $B$. 
    Group $A$ has one agent, denoted as $a$, with $f(\{a\})=\frac{1}{2}$ and $c_a=\frac{1}{2M}$.
    Group $B$ has $n-1$ identical agents $b,\ldots,b$ with $f(\{b\})=\frac{1}{2(n-1)}$ and $c_b=\frac{\epsilon}{2(n-1)^2}$.
    Without constraints, agent $a$ requires payment at least $\frac{1}{M}$ and each agent $b$ requires at least $\frac{\epsilon}{n-1}$ to exert effort.
    With arbitrarily small $\epsilon$, it is not hard to see that in the optimal unconstrained contract, the principal incentivizes all agents to achieve utility $\OPT=1-\frac{1}{M}-\epsilon$.

    We now consider $n^\delta$-ND contracts. If the contract incentivizes only agents in group $B$ (reward zero for $a$ and reward $\frac{\epsilon}{2(n-1)}$ for each $b$), the principal has utility $\frac{1-\epsilon}{2}$.
    If the $n^\delta$-ND contract incentivizes $a$ plus a number $t\geq 0$ of agents $b$,
    then the optimal way is to set payment $\frac{1}{M}$ to $a$ and $\frac{1}{n^\delta M}$ to each incentivized $b$, as $\frac{1}{n^\delta M} > \frac{\epsilon}{(n-1)}$ ($M \ll n^{1-\delta}$ and $\epsilon<1$).
    This contract yields utility
    \[
    \left( \frac{1}{2}+\frac{t}{2(n-1)} \right) \left( 1- \frac{1}{M} - \frac{t}{n^\delta M} \right) \triangleq h(t).
    \]
    By taking derivatives, the function $h(t)$ is strictly decreasing in $t$. Thus, the principal’s optimal utility under $n^\delta$-ND contracts is $\frac{1}{2}-\frac{1}{M}$, achieved by incentivizing agent $a$ only.
    Therefore, the ratio between the optimal utility under $n^\delta$-ND and unconstrained contracts is
    \[
    \frac{\OPT_{ND}}{\OPT}=\frac{\frac{1}{2}-\frac{1}{M}}{1-\frac{1}{M}-\epsilon} < \frac{1}{2}+\epsilon, \text{ as } M>3+\frac{1}{\epsilon},
    \]
    which completes the proof.
\end{proof}


\noindent\textbf{Lemma 9.} \textit{There exists an instance where the principal's optimal utility under $n$-ND contracts is no greater than $0.629$ fraction of the optimal utility under unconstrained contracts.}

\begin{proof}
    Let $n$ be an even number and $n \gg 1$. Consider an instance with additive $f$ and $n$ agents.
    Agents are classified into groups $A$ and $B$, where $A$ contains one agent denoted by $a$, and group $B$ contains $n-1$ identical agents $b,\ldots,b$.
    Let $f(\{a\})=\frac{\sqrt{2}}{4}$ and $c_a=\frac{\sqrt{2}-1}{4}$, and hence, agent $a$ should be paid at least $1-\frac{1}{\sqrt{2}}$ to exert effort.
    For each $b$, let $f(\{b\})=\frac{1}{4(n-1)}$ and $c_b=\frac{\epsilon}{4(n-1)^2}$, where $\epsilon>0$ is arbitrarily small. 
    Thus, to incentivize each $b$, payment at least $\frac{\epsilon}{n-1}$ is required.
    It is not hard to see that the optimal unconstrained contract is to incentivize all of the $n$ agents with payment $1-\frac{1}{\sqrt{2}}$ to $a$ and $\frac{\epsilon}{n-1}$ for each $b$, and gives utility $\OPT=(\frac{\sqrt{2}+1}{4})(\frac{1}{\sqrt{2}}-\epsilon)$ to the principal.

    For an $n$-ND contract, if it incentives only (all) agents in $B$, the principal's utility is $\frac{1-\epsilon}{4}$. 
    If the $n$-ND contract tends to incentivize $a$ plus a number $t\geq 0$ of agents $b$, then the optimal way is to pay $1-\frac{1}{\sqrt{a}}$ to $a$ and $\frac{\sqrt{2}-1}{n\sqrt{2}}$ to each chosen $b$, yielding a utility of
    \[
    \left(\frac{\sqrt{2}}{4}+\frac{t}{4(n-1)}\right)\left(\frac{1}{\sqrt{2}}-\frac{t(\sqrt{2}-1)}{n\sqrt{2}}\right)\triangleq h(t).
    \]
    By taking derivatives, one can verify that $h(t)$ is maximized at $t=\frac{n}{2}$, assuming only integers are considered. 
    Then the $n$-ND contract that incentivizes $a$ and half of $b$ gives the optimal utility
    \[
    \OPT_{ND}= \frac{10-\sqrt{2}}{32} + \frac{1}{8n-8}.
    \]
    Since $n\gg 1$, we have
    \[
    \frac{\OPT_{ND}}{\OPT} \approx  \frac{11-6\sqrt{2}}{4}\approx 0.629,
    \]
    completing the proof.
\end{proof}

\section{Missing Proofs from Section 5}\label{a5}
\textbf{Theorem 3.}\textit{ When there are two agents and $f$ is submodular, the $\PoND(\beta)$ equals $(1+\frac{1}{\sqrt{\beta+1}})$ for any $\beta \geq 1$, and this bound is tight.}

\begin{proof}
    In this proof, for simplicity, we write $f(1)$, $f(2)$ and $f(1,2)$ as $f(\{1\})$, $f(\{2\})$ and $f(\{1,2\})$ respectively. 
    The costs of agents 1 and 2 are $c_1$ and $c_2$.
    For the instance where the optimal unconstrained contract incentivizes only one agent, such a contract is also $\beta$-ND.
    Hence, it suffices to consider only the instances where the optimal unconstrained contract incentivizes both agents.

    Define $\alpha'_1:=\frac{c_1}{f(1 \,| \,\{2\})}$ and $\alpha'_2:=\frac{c_2}{f(2 \,| \,\{1\})}$, the reward for agent 1 and 2 in the optimal unconstrained contract.
    Without loss of generality, assume $0\leq \alpha'_1\leq \alpha'_2<1$.
    If $\alpha'_1\geq \frac{\alpha'_2}{\beta}$, the optimal unconstrained contract satisfies the $\beta$-ND requirement.
    Thus, we focus on case where $\alpha'_1<\frac{\alpha'_2}{\beta}$.

    We consider the three options of the $\beta$-ND contract: incentivizing only $a$; incentivizing only $b$; or incentivizing both agents.
    We will show that at least one of the options achieve the desired ratio, i.e.,
    \begin{align*}
        &\max\bigg\{ f(1)(1-\frac{c_1}{f(1)}), f(2)(1-\frac{c_2}{f(2)}), \\
        &\qquad \quad f(1,2)(1-\frac{\alpha'_2}{\beta}-\alpha'_2)\bigg\} \\
        \ge &\frac{\sqrt{\beta+1}}{\sqrt{\beta+1}+1}f(1,2)(1-\alpha'_1-\alpha'_2).    
    \end{align*}

    Define $x:=\frac{f(1)}{f(1)+f(2)}$. Since $f$ is submodular, $x \leq \frac{f(1)}{f(\{1,2\})}$ and $1-x\leq \frac{f(2)}{f(\{1,2\})}$.
    Moreover, $\alpha'_1\geq \frac{c_1}{f(1)}$ and $\alpha'_2\geq \frac{c_2}{f(2)}$.
    To show the existence of the desired ND contract, it suffices to show
    \begin{align}\label{2a:obj-1}
        \begin{split}
            &\max\left\{x(1-\alpha'_1), (1-x)(1-\alpha'_2), 1-\frac{\alpha'_2}{\beta}-\alpha'_2 \right\} \\
            \ge &\frac{\sqrt{\beta+1}}{\sqrt{\beta+1}+1}(1-\alpha'_1-\alpha'_2).
        \end{split}
    \end{align}

    Assume by contradiction that Inequality~(\ref{2a:obj-1}) does not hold. Substituting $1-\alpha'_1, 1-\alpha'_2$ by $u,v$ gives us: 
    \begin{subequations}
    \begin{align}
    x u &< \frac{\sqrt{\beta+1}}{\sqrt{\beta+1}+1}(u+v-1), \label{a}\\
    (1-x)v &< \frac{\sqrt{\beta+1}}{\sqrt{\beta+1}+1}(u+v-1), \label{b}\\
    1-\frac{\beta+1}{\beta}(1-v) &< \frac{\sqrt{\beta+1}}{\sqrt{\beta+1}+1}(u+v-1). \label{c}
    \end{align}
    \end{subequations}
    Multiplying both sides of Inequality~(\ref{a}) by $\frac{1}{u}$ and~(\ref{b}) by $\frac{1}{v}$, and then adding them, yields
    \begin{align}\label{contradict}
        \frac{(u+v-1)(u+v)}{uv}>1+\frac{1}{\sqrt{\beta+1}}. 
    \end{align}
    Denote by $H(u,v)$ the expression on the left hand side of the above inequality.
    We now check the monotonicity of function $H(u,v)$ with respect to $u, v$; recall that $\alpha'_1,\alpha'_2\in[0,1)$ and thus $u, v\in (0,1]$.
    By taking the partial derivatives,
    \[
    \frac{\partial H(u,v)}{\partial u}=\frac{u^2-v(v-1)}{u^2 v}>0 \text{ for all } u \in (0,1].
    \]
    Symmetrically, we have $\frac{\partial H(u,v)}{\partial v}>0$ for all $v\in (0,1]$.
    
    We finally show that the maximum value of $H(u,v)$ with $0\leq v\le u<1$ and Inequality~(\ref{c}) is less than $1+\frac{1}{\sqrt{\beta+1}}$, which leads to a contradiction to Inequality~(\ref{contradict}). This indicates that Inequalities~(\ref{a}), (\ref{b}), and (\ref{c}) can never hold simultaneously, which derives the desired contradiction.

    To establish the upper bound on $H(u,v)$, we first derive an upper bound on $v$ from Inequality~(\ref{c}).
    Rearrange it gives 
    \[
    (\beta+1-\frac{\beta\cdot\sqrt{\beta+1}}{\sqrt{\beta+1}+1})v < \frac{\beta\cdot\sqrt{\beta+1}}{\sqrt{\beta+1}+1} u-(\frac{\beta\cdot\sqrt{\beta+1}}{\sqrt{\beta+1}+1}-1).
    \]
    Since $u\leq 1$ and the coefficient of $v$ is non-negative, we have 
    \[
    v\leq \frac{1}{\beta+1-\frac{\beta\cdot\sqrt{\beta+1}}{\sqrt{\beta+1}+1}}=\frac{\sqrt{\beta+1}+1}{\sqrt{\beta+1}+\beta+1}=\frac{1}{\sqrt{\beta+1}}.
    \]
    Given the upper bounds of $u,v$ and the monotonicity of $H(u,v)$, we have
    \[
    H(u,v)\leq H(1,\frac{1}{\sqrt{\beta+1}}) = 1+\frac{1}{\sqrt{\beta+1}},
    \]
    contradicting to Inequality~(\ref{contradict}).
    Therefore, Inequality~(\ref{2a:obj-1}) is proved since (\ref{a}),(\ref{b}),(\ref{c}) cannot hold simultaneously.

    For the tightness let us consider the instance with $f(1)=\frac{1}{2}, c_1=\frac{1}{2}-\frac{1}{2\sqrt{\beta+1}}, f(2)=\frac{1}{2\sqrt{\beta+1}}, c_2=\frac{\epsilon}{2\sqrt{\beta+1}}$, where $\epsilon>0$ is arbitrarily small.
    The optimal unconstrained contract sets payment $1-\frac{1}{\sqrt{\beta+1}}$ to agent 1 and $\epsilon$ to agent 2, which yields the principal's utility $\frac{\sqrt{\beta+1}+1}{2(\beta+1)} - \epsilon(\frac{1}{2}+\frac{1}{2\sqrt{\beta+1}})$ for the principal.
    For $\beta$-ND contracts, (1) incentivizing only agent 1 with payment $1-\frac{1}{\sqrt{\beta+1}}$ yields utility $\frac{1}{2\sqrt{\beta+1}}$, 
    (2) incentivizing only agent 2 with payment $\epsilon$ yields utility $\frac{1}{2\sqrt{\beta+1}}(1-\epsilon)$, 
    and (3) incentivizing both agents, with payment $1-\frac{1}{\sqrt{\beta+1}}$ to agent 1 and $\frac{1}{\beta}(1-\frac{1}{\sqrt{\beta+1}})$ to agent 2, yields utility $\frac{1}{2\sqrt{\beta+1}}$.
    Therefore, the optimal utility under $\beta$-ND contracts is $\frac{1}{2\sqrt{\beta+1}}$, and thus,
    \[
    \frac{\OPT}{\OPT_{ND}}=1+\frac{1}{\sqrt{\beta+1}} \implies \beta\text{-}\PoND \leq 1+\frac{1}{\sqrt{\beta+1}},
    \]
    completing the proof.
\end{proof}

\end{document}